\documentclass[preprint,12pt,authoryear]{elsarticle}

\usepackage{epsfig}
\usepackage{subfigure}
\usepackage{amssymb}
\usepackage{amsmath}
\usepackage{float}
\DeclareMathOperator{\rank}{rank}
\DeclareMathOperator{\nullity}{null}
\usepackage{amsthm}
\usepackage{amsbsy}
\usepackage{amsfonts}
\usepackage{xcolor}
\usepackage{booktabs}
\usepackage[hidelinks]{hyperref}

\newtheorem{theorem}{Theorem}

\newtheorem{proposition}[theorem]{Proposition}
\newtheorem{property}[theorem]{Property}

\newtheorem{definition}{Definition}

\newtheorem{remark}{Remark}

\begin{document}
	
	\begin{frontmatter}
	
	
	\title{
		Lower spectrum of financial correlation matrices: a new perspective on market synchronization  
	}
	
	
	\author[1]{Rosanna Grassi%
	}
	\ead{rosanna.grassi@unimib.it}
	
	\author[2]{Caterina Pastorino\corref{cor1}}
	\ead{caterina.pastorino@unimib.it}
	
	\author[1]{Pierpaolo Uberti}
	\ead{pierpaolo.uberti@unimib.it}
	
	\cortext[cor1]{Corresponding author}
	
	\affiliation[1]{organization={Department of Statistics and Quantitative Methods, University of Milano-Bicocca},
		city={Milan},
		country={Italy}}
	
	\affiliation[2]{organization={Department of Economics, Management and Statistics, University of Milano-Bicocca},
		city={Milan},
		country={Italy}}

	\begin{abstract}
{
In this paper we investigate the information content of the lower part of the spectrum of financial correlation matrices, as a source of information on market synchronization. In a financial context, a classical application of Principal Component Analysis and Random Matrix Theory identifies the largest eigenvalues as indicators of dominant market factors and synchronization patterns. We complement this perspective by showing that the smallest eigenvalues also contain relevant information about the effective structure of financial markets.
The paper presents the methodological proposal and validates its effectiveness through comprehensive real data experiments in both descriptive and predictive settings. 

}
\end{abstract}
			
			
			
	
	\end{frontmatter}
	
	
	
\section{Introduction}\label{intro}

Financial markets are characterized by alternating periods in which assets evolve independently and periods in which their dynamics become increasingly synchronized. As synchronization increases, the effective dimensionality of the market decreases, reducing the number of independent opportunities for diversification. These structural changes become particularly evident during episodes of severe market distress. The 2008 subprime global crisis clearly showed the effects of a large-scale market synchronization, and this reinforced the need for quantitative tools capable of characterizing its evolution. 
Driven by the growing interconnectedness of global financial markets (\cite{26.diebold2014network}), asset returns have become increasingly prone to collective co-movement (\cite{lavin2021network, magner2022modeling}). When market synchronization intensifies, reduced diversification accelerates shock propagation across the market, driving financial contagion (\cite{nguyen2026community}) and systemic crises (\cite{alexandre2021drivers, wei2025self}). 
These recurring episodes have stimulated an extensive body of literature on systemic risk assessment and monitoring has emerged; for comprehensive reviews, we refer to \cite{14.Bisias, 45.rodriguez2013systemic, 47.silva2017analysis,benoit2017risks,bardoscia2021physics}. Within this wide context, understanding how market synchronization evolves has become an important research direction.

Since this phenomenon is related to the correlation structure among assets, the spectral analysis of correlation matrices has become one of the main quantitative approaches to investigate the collective market dynamics. A large part of the literature investigated the relevant economic information contained in the largest eigenvalues. In this direction, an important role is played by Random Matrix Theory (RMT), used as a benchmark to distinguish meaningful spectral information from noise in financial correlation matrices (\cite{Laloux1999}). Other studies mainly focused on the largest eigenvalues, interpreting them as dominant market factors or market mode (\cite{plerou2002random, junior2012correlation, Garlaschelli, lux2020analysis}). 
Moving from this interpretation, the information contained by the higher eigenvalues has been exploited to construct several quantitative indicators to track changes in  collective market behaviors.
Among these approaches, Principal Component Analysis (PCA) has become one of the most widely adopted frameworks. Indeed, the first principal components are interpreted as common market factors driving co-movements of asset returns. These indicators have been applied in different contexts, including the measurement of systemic risk (\cite{13.Billio, 53.zheng2012changes,52.zhang2020global}).
These approaches share a common aspect, that is they detect economically relevant information almost exclusively focusing on the upper part of the spectrum. In contrast, the opposite side of the spectrum, has received considerably less attention. 

The aim of this work is to investigate the information embedded in the lower part of the spectrum, that is the smallest eigenvalues and their relation with diversification collapse and effective dimensionality reduction. In our framework, Random Matrix Theory plays a different role. Indeed, rather than being used as a benchmark to identify informative eigenvalues of the correlation matrix, it provides a definitive result on their distribution. If the columns of a given matrix are the realizations of $n$ independent and identical distributed
random variables, then its eigenvalues follow the so-called Marchenko-Pastur
distribution (\cite{marchenko1967distribution}). Therefore, the eigenvalues of
the correlation matrix range in a suitable interval, and the presence
of eigenvalues outside the theoretical Marchenko–Pastur support indicates departures from a random-correlation benchmark and signals the presence of
non-trivial dependence structures in the data. 
In the financial context, when many securities move in a similar way, the number of effectively independent investments decreases and
the diversification opportunities reduce. 

To correctly interpret the information included in the lower part of the spectrum, we rely on the concepts of numerical nullity and numerical rank of the correlation  matrix. Indeed, we show that the extreme parts of the spectrum are strongly related to both numerical nullity and rank of the correlation matrix. This provides a natural bridge between the mathematical properties and its financial interpretation.  

It is worth point out that, by RMT, our approach avoids an arbitrary definition of the lower part of the spectrum. The threshold is theoretically determined by the Marchenko–Pastur distribution, and the corresponding number of small eigenvalues follows naturally from this choice.
According to this pattern, we introduce a lower-spectrum indicator 
based on the number of eigenvalues below the lower Marchenko–Pastur bound. From a geometric perspective, this corresponds to a correlation matrix approaching a low dimensional structure.
The definition arises in a natural way and we prove that the proposed indicator belongs to the class of the so-called Proper Measures of Connectedness (\cite{maggi2020proper}).

To show the effectiveness of the proposed indicator, we test it on three financial datasets, different in both assets type and in size. The first one is the sectoral portfolios constructed by the  S\&P500 index following the global industry classification standard (GICS). In the second database the assets are the stocks constituting the whole index. 
The scope is evaluating the proposed indicator across different scales of market aggregation. The third dataset contains the assets of the Nikkei index, used to support the robustness of our findings.  
To allow a correct applicability of the theoretical conditions, the dimensionality of the datasets is reduced by a clustering procedure. For comparison purposes, we also report the behavior of an upper-spectrum indicator, based on eigenvalues above the upper Marchenko--Pastur bound and closely related to the standard PCA/RMT literature.

Findings definitely confirm the solidity of the proposed theoretical methodology. Indeed, the indicator captures episodes of stress in both the in-sample and out-of-sample analyses. In particular, it reveals interesting predictive capabilities, proving to be more informative than measures constructed exclusively on the upper part of the spectrum.
%
The paper is organized as follows: Section \ref{theo} define the theoretical framework; Subection \ref{measures} provides the definitions of the extreme spectrum indicators, together with the theoretical properties. Section \ref{empirical} is dedicated to the in-sample and out-of-sample empirical applications on real financial data. Finally, Section \ref{conclusions} draws the conclusions.

\section{The theoretical framework}\label{theo}

\subsection{Theoretical background}

This section introduces the theoretical framework underlying the proposed lower-spectrum indicator. We recall the axiomatic framework of Proper Measures of Connectedness (\cite{maggi2020proper}) and Marchenko–Pastur distribution (\cite{marchenko1967distribution}), and the definition of numerical rank and numerical nullity of a matrix (\cite{golub2013matrix}).

Let us denote by ${\rm Mat}_{T \times n}(\mathbb R)$, with $T \geq n \geq 2$, the set of $T \times n$ real matrices. Given $A \in {\rm Mat}_{T \times n}(\mathbb R)$, the {\em rank} of $A$, $\rank(A)$, is the dimension of its column space (or
row space), and the {\em nullity} of $A$, $\nullity(A)$, is the dimension of the kernel of $A$. By the rank-nullity theorem, $\nullity(A)=n-\rank(A)$. Using a standard notation, $\sigma_1(A)\ge \ldots \ge \sigma_n(A) \ge 0$ are the singular values of $A$.  The rank is equal to the number of nonzero singular values of $A$. 
We denote by ${\mathcal M}_{T \times n}$ the subset of ${\rm Mat}_{T \times n}(\mathbb R)$ containing 
all the matrices with full rank, then equal to $n$. \\ 
Let  $A \in {\mathcal M}_{T \times n}$ be the matrix of returns of $n$ risky assets, where the $i^{th}$ column $A^i \in \mathbb R^T$ contains the time series of $T$ returns of the asset $i$, $i=1,\dots, n$.

\begin{definition}[Proper Measure of Connectedness  \citep{maggi2020proper}]\label{properMeasures}
A real-valued function $C: {\mathcal M}_{T \times n} \rightarrow \mathbb R$ is a \emph{Proper Measure of Connectedness (PMC)} 
if it satisfies the following properties \ref{property1}, \ref{property2}, \ref{property3} 				and \ref{property4}:
\end{definition}

\begin{property}\label{property1}
$C(A) \geq 0$, for any $A \in {\mathcal M}_{T \times n}$.
\end{property}
			
\begin{property}\label{property2}
$C(A)$ is invariant for any permutation of the columns of $A \in {\mathcal M}_{T \times n}$.
\end{property}
	
\begin{property}\label{property3}
$C(A)>C(B)$ if and only if $C(\alpha A)>C(\alpha B)$, for any $A, B \in {\mathcal M}_{T \times n}$ and for any real number $\alpha>0$.
\end{property}
			
\begin{property}\label{property4} 
Let $A^1, A^2, A^3 \in {\mathbb R}^{T}$ with \(\|A^2\|=\|A^3\|\) and \(\langle\mathbf{1},A^2\rangle=\langle\mathbf{1},A^3\rangle=0\), where \(\mathbf{1}\) is the vector of ones, $\langle\cdot,\cdot \rangle$ represents the scalar product and $\|\cdot\|$ is the Euclidean norm. 
Let $A_{ij} = (A^i | A^j) \in {\mathcal M}_{T \times 2}$, and $\rho_{ij}$ be the linear correlation coefficient between $A^i$ and $A^j$, $i,j=1,2,3.$
If $|\rho_{12}| \ge |\rho_{13}|$ then $C(A_{12}) \ge C(A_{13})$.
\end{property}
	
We denote by ${\mathcal S}_{T \times n}$ the subset of ${\mathcal M}_{T \times n}$ that contains all the $T \times n$ full rank matrices with standardized columns\footnote{ $S = \left[s_{i,j}\right] \in {\mathcal S}_{T \times n}$ is such that $\frac{1}{T}\sum_{i = 1}^T s_{i,j} = 0$ and $\frac{1}{T}\sum_{i = 1}^T s^2_{i,j} = 1$, for $j = 1, \ldots, n$.}. Notice that, given a matrix of returns $A \in {\mathcal M}_{T \times n}$, it is always possible to construct the correspondent version $S_A \in {\mathcal S}_{T \times n}$ by standardizing the columns of $A$. Let us denote by $\frac{1}{T} S_A'S_A$ the correlation matrix between the assets, where $S'_A$ is the transpose of $S_A$.
Without loss of generality, we directly work with a matrix in ${\mathcal S}_{T \times n}$. 

A classical result in RMT is the following:
\begin{theorem}[Marchenko-Pastur distribution]\label{MPdist} Given $S \in {\mathcal S}_{T \times n}$, if each column $S^i$, for $i = 1, \ldots, n$, contains $T$ realizations of independent and identically distributed random variables, then for $T \rightarrow +\infty$ and $n \rightarrow +\infty$, with $\frac{n}{T} \rightarrow \lambda \in (0,1)$ it holds that
\[
\lambda_1, \lambda_2, \ldots, \lambda_n \in (\lambda_{min}, \lambda_{max})
\]
where $\lambda_{min} = (1-\sqrt{\lambda})^2$, $\lambda_{max} = (1+\sqrt{\lambda})^2$ and $\lambda_1 \geq \lambda_2 \geq \ldots \geq \lambda_n$ are the eigenvalues of $\frac{1}{T} S'S$, written in non-increasing order.  
\end{theorem}

The theorem states that, if the columns of $S$ are independent and identical distributed, then its eigenvalues range in the interval $(\lambda_{min}, \lambda_{max})$. We exploit this result in Section \ref{empirical}, where, starting from the matrix of returns $A \in {\mathcal M}_{T \times n}$, we construct the correlation matrix $\frac{1}{T} S_A'S_A$, with $S_A \in {\mathcal S}_{T \times n}$.


We underline that 
the assumption $\rank(A)=n$ guarantees the linear independence of the random
variables that generate the columns of $S_A$. Nevertheless, the linear independence and the fact that the columns of $S_A$ all have equal mean and variance, do not imply that the random variables that generate the columns of $S_A$ are independent and identically distributed.   
However, the application of Theorem \ref{MPdist} under these general hypotheses is commonly accepted in the literature (see, among others \cite{yaskov2016, Garlaschelli, lux2020analysis,guerini2023synchronization}), and the Marchenko–Pastur theorem provides a benchmark for identifying deviations from random correlation structures. \\
However, to interpret the information included in the lower part of the spectrum, a link with the algebraic properties of the correlation matrix is needed. This bridge is naturally provided through the concepts of numerical nullity and numerical rank of a matrix.

Given a parameter $\varepsilon > 0$, the {\it numerical $\varepsilon-$ nullity} and the {\it numerical $\varepsilon-$ rank} of a matrix are defined as follows (see \cite{golub2013matrix}).

\begin{definition}\label{numRank}
Let $\varepsilon>0$; the {\it numerical $\varepsilon-$ nullity} and {\it numerical $\varepsilon-$ rank} of a matrix $A \in {\rm Mat}_{T \times n}(\mathbb R)$, denoted by $N_\varepsilon(A)$ and $R_\varepsilon(A)$ respectively, are defined by
\begin{eqnarray*}
N_\varepsilon(A) &:=& \max_{B \in {\rm Mat}_{T \times n}(\mathbb R)} \{\nullity(B) \::\: \|B-A\| \le \varepsilon\}\\
R_\varepsilon(A) &:=& \min_{B \in {\rm Mat}_{T \times n}(\mathbb R)} \{\rank(B) \::\: \|B-A\| > \varepsilon\}.
\end{eqnarray*}

Equivalently, using the singular values of $A$, $N_\varepsilon(A)$ and $R_\varepsilon(A)$ are defined by:
\begin{eqnarray*}
N_\varepsilon(A) &:=& \#\{k \in \{1,\ldots,n\} \:|\: \sigma_k(A) \le \varepsilon\}\\
R_\varepsilon(A) &:=& \#\{k \in \{1,\ldots,n\} \:|\: \sigma_k(A) > \varepsilon\}.
\end{eqnarray*}
\end{definition}

The numerical nullity (the numerical rank) of a matrix $A$ is the number of its singular values that are lower (greater) than a certain threshold parameter $\varepsilon >0$. The extreme parts of the spectrum are then strongly related to both numerical nullity and rank of the correlation matrix.

The role of the parameter $\varepsilon$ deserves a particular discussion.
In general, this parameter identifies the tolerance that discriminates the numbers that are indistinguishable from zero from the ones that are significantly positive. In standard numerical applications, $\varepsilon$ is equal to the machine tolerance, this parameter is arbitrary and needs to be decided. 

One of the main advantages of our proposal is that 
RMT avoid any arbitrariness on the choice of $\varepsilon$.
Indeed, our approach replaces an arbitrary tolerance parameter with a theoretically motivated value derived from RMT, as highlighted in the next section. 

\subsection{Lower-spectrum and upper-spectrum indicators}\label{measures}

As previously said, the random correlation structures are bounded in a theoretical interval by Theorem \ref{MPdist}. Thus, it is worth investigating the information contained in the eigenvalues falling out of this interval. Although the upper side of the spectrum has been extensively studied in the literature, the lower side has not been sufficiently investigated yet. Both sides have a mathematical interpretation through the concepts of numerical $\varepsilon$-rank and $\varepsilon$-numerical nullity. Therefore, we define the following indicators as:

	\begin{definition}\label{epsNullity}
	Let 
	$S_A \in {\mathcal S}_{T \times n}$ and, let $\varepsilon^- = 1-\sqrt{\lambda}$ and $\varepsilon^+ = 1+\sqrt{\lambda}$, where $\lambda$ is defined as in Theorem \ref{MPdist}, then
	\begin{eqnarray}
	m^-=m^-(S_A) := \# \{k \in \{1,\ldots,n\} \: |\: \sigma_k(S_A) \leq \varepsilon^- \}.
	\end{eqnarray}
and
\begin{eqnarray}\label{ktau}
	m^+ =m^+(S_A):= \# \{k \in \{1,\ldots,n\} \: |\: \sigma_k(S_A) > \varepsilon^+\}
\end{eqnarray}

	\end{definition}

Notice that $\varepsilon^-$ and $\varepsilon^+$ in Definition \ref{epsNullity} are not arbitrarily chosen, but they depend on the Marchenko-Pastur distribution	\footnote{Different choices of the parameters 
$\lambda_{min}, \lambda_{max}$ following eventual alternative criteria are possible, but this is beyond the scope of the present research. } 
			
\begin{remark}\label{REM_def3}
As the square matrix $\frac{1}{T} S_A'S_A$ is constructed starting from $A$, the definition of the measures $m^-$ and $m^+$ can be formulated equivalently in terms of the number of the eigenvalues of $\frac{1}{T} S_A'S_A$. Thus,
$m^-$ counts the eigenvalues of the correlation matrix that are smaller than $\lambda_{min}$, while $m^+$ counts the eigenvalues of the correlation matrix larger than $\lambda_{max}$.
\end{remark}

Definitions \ref{numRank} and \ref{epsNullity} are strictly related.
From Definition \ref{epsNullity} 
it is straightforward to observe that, if $\varepsilon^- = \varepsilon^+= \varepsilon$, then $m^+$ and $m^-$ would be respectively the numerical rank and the numerical nullity of $S_A$ of Definition \ref{numRank}. 
Actually, this occurrence is impossible in the framework of RMT since $(\varepsilon^-)^2 = \lambda_{min} < \lambda_{max} = (\varepsilon^+)^2$. 

This allows us giving to the proposed measures a deeper interpretation, as highlighted in the following.
Both indicators identify deviations from the Marchenko--Pastur interval at opposite sides of the spectrum. 
We focus on the low part of the spectrum: the measure $m^-$ counts the number of eigenvalues below the lower Marchenko--Pastur bound. According to the interpretation of the numerical nullity, the presence of several small eigenvalues indicates that 
the correlation matrix approaches a lower-rank structure. As a consequence, groups of assets become increasingly linearly dependent. From a portfolio perspective, this corresponds to a reduction of the effective diversification opportunities during periods of market distress. Note that the upper-spectrum indicator $m^+$ is closely related to the standard PCA and RMT literature, where large eigenvalues are interpreted as dominant market factors, or a market mode, as defined in \cite{Garlaschelli}. In this paper we focus on the lower part of the spectrum, and use $m^+$ mainly as a benchmark and a comparison tool. 	
					
The following proposition shows that both indicators satisfy the axioms of a Proper Measure of Connectedness. This result allows to classify $m^-$ within the connectedness-measurement literature. 

\begin{proposition}\label{propPMC}
				$m^+$ and $m^-$, defined as in Definition \ref{epsNullity}, satisfy the properties of a PMC, $\forall S_A \in {\mathcal S}_{T \times n}$.
\end{proposition}
			
\begin{proof}
				Let $S_A\in {\mathcal S}_{T \times n}$ with singular values $\sigma_1(S_A) \ge \ldots \ge \sigma_n(S_A)>0$. Property \ref{property1} immediately follows from Definition \ref{epsNullity}. 
				
				Since the singular values of a matrix are independent from any column permutation, then $m^+$ and $m^-$ satisfy Property \ref{property2}.
				
				To prove Property \ref{property3}, it is sufficient to note that $S_A = S_{\alpha A}$ for each $\alpha \in \mathbb R^+$. A positive rescaling of all the entries of a matrix does not affect the singular values.

				We now prove Property \ref{property4}. By Remark \eqref{REM_def3}, we can refer to the definition in terms of eigenvalues instead of singular values. Consider the returns $A^1, A^2, A^3 \in {\mathbb R}^{T}$  such that  \(\|A^2\|=\|A^3\|\) and \(\langle\mathbf{1},A^2\rangle=\langle\mathbf{1},A^3\rangle=0\) and the correlation matrices:
				\[
				\frac{1}{T}S_{A_{12}}'S_{A_{12}} = \left[\begin{array}{cc} 1 & \rho_ {12}\\ 
					\rho_{12} & 1 \end{array}\right] \quad \quad
				\frac{1}{T}S_{A_{13}}'S_{A_{13}}  = \left[\begin{array}{cc} 1 & \rho_{13}\\ 
					\rho_{13} & 1 \end{array}\right].
				\]	
				Under these assumptions, if $|\rho_{12}| \ge |\rho_{13}|$ we have to prove the following inequalities: $m^-(S_{A_{12}})\geq m^-(S_{A_{13}})$ and $m^+(S_{A_{12}})\geq m^+(S_{A_{13}}).$ \\
				Observe that the eigenvalues of $\frac{1}{T}S_{A_{12}}'S_{A_{12}}$ are $ 1 \pm |\rho_{12}|$ 	
				and those of $\frac{1}{T}S_{A_{13}}'S_{A_{13}}$ are $ 1 \pm |\rho_{13}|$. Moreover:
				
				\[
				1 + |\rho_{12}| \geq 1 + |\rho_{13}| > 1 - |\rho_{13}| \geq 1 - |\rho_{12}|. 
				\]
				Focusing on the measure $m^-$, the threshold $\varepsilon^- $ is equal to $ 1-\sqrt{\frac{2}{T}}$ and, thanks to the assumption $T > n=2$:
				\[
				1 + |\rho_{12}| \geq 1 + |\rho_{13}| > \left(1-\sqrt{\frac{2}{T}}\right)^2=(\varepsilon^-)^2.
				\]
				Thus, we need to compare the threshold $\varepsilon^-$ only with the two smallest eigenvalues. 
				Only three cases are possible:

				\[
				\left\{\begin{array}{lcl}
					1 - |\rho_{12}|>\left(1-\sqrt{\frac{2}{T}}\right)^2 & \textrm{then} & 
					m^-(S_{A_{12}}) = m^-(S_{A_{13})}) = 0   \\ 
					1 - |\rho_{12}| \leq \left(1-\sqrt{\frac{2}{T}}\right)^2 \leq 1 - |\rho_{13}| & \textrm{then} & 
					m^-(S_{A_{12}}) = 1, \quad m^-(S_{A_{13}}) = 0   \\ 
					1 - |\rho_{13}| < \left(1-\sqrt{\frac{2}{T}}\right)^2 & \textrm{then} & 
					m^-(S_{A_{12}}) = m^-(S_{A_{13}}) = 1    
				\end{array}\right..
				\]

				In all the cases $m^-(S_{A_{12}}) \geq m^-(S_{A_{13}})$, concluding that $m^-$ verifies Property \ref{property4}.
				
				An analogous argument can be used for $m^+$. In this case, $\varepsilon^+ = 1+\sqrt{\frac{2}{T}}$ and in all the cases $m^+(S_{A_{12}}) \geq{} m^+(S_{A_{13}})$. This concludes the proof.

\end{proof}

\section{Empirical Analysis}\label{empirical}
			
In this section, we evaluate the performance of the proposed lower-spectrum indicator on real data in both descriptive and predictive settings.
The first subsection provides a brief description of the experiment and data. Then, we present the in-sample and out-of-sample experiments, respectively.

\subsection{Empirical design and datasets}\label{data}

The analysis is performed in a rolling window framework. We start from the matrix $A \in {\mathcal M}_{T \times n}$ containing the time series of $T$ returns of $n$ risky assets and we set the length of the rolling window equal to $w$ (with $w >n$). Then, we compute the eigenvalues of the correlation matrix constructed with respect to the first $w \times n$ entries of $S_A$. More precisely, we extract the $w$ rows from $A$, then we standardize its columns and construct the correlation matrix.\\
For each rolling window, the indicator $m^-$ is computed according to Definition \ref{epsNullity}. Since the window length $w$ and the number of assets $n$ are given, the ratio $\lambda=\frac{n}{w}$ and, consequently, the values of $\lambda_{max}$ and $\lambda_{min}$ remain constant throughout the experiment.  This procedure is repeated shifting the rolling window one day ahead. 
The resulting time series of the measure is then compared with the corresponding market index to investigate how these changes capture different market conditions. Whenever relevant, the corresponding upper-spectrum indicator $m^+$ is also reported as a benchmark.  
The choice of the parameter $w$ is crucial in our framework for many reasons. At first, let us recall that $w > n$ is the necessary condition to apply RMT. Moreover, it identifies a referring time horizon for the duration of the systemic event. A short $w$ permits to compute the measures on the most recent data, the ones that better describe actual market conditions. Moreover, a measure calculated on a short window is expected to be reactive and helpful also in a predictive framework. The main drawback of working with a reactive measure is the potential high number a false signals.
On the opposite, if $w$ is long, the measure is stable, but it can loose the capacity of timing the market, not just in a predictive, but also in a descriptive context. 

We perform our analysis on three different datasets:	
\begin{itemize} 
		\item \textit{SPsector} contains the daily logarithmic returns from 01-03-1990 to 09-17-2020 of the S\&P500 index and its $n=10$ sector sub-indexes based on GICS (Global Industry Classification Standard): financials, information technology, telecommunication services, health care, industrials, consumer discretionary, energy, consumer staples, utilities, materials. 
		\item \textit{SPstocks} includes the daily logarithmic returns of the S\&P500 index and $n=377$ of its constituents from 01-05-2005 to 09-17-2020.
		\item \textit{Nikkei} contains the daily logarithmic returns of the Nikkei225 index and $n = 199 $ of its constituents from 01-05-2005  to 09-17-2020.
\end{itemize}

The choice of the datasets allows us to evaluate the proposed indicator across different scales of market aggregation.
In particular, at the sectoral-level, we investigate the behavior of lower-spectrum indicator across various economic sectors, whereas the stock-level datasets assess whether the proposed indicator also remains informative in higher-dimensional settings. It is worth to note that
the number of stocks in SPstocks and Nikkei is smaller than the nominal number of the index constituents because we restrict the analysis to the stocks with complete time series on the chosen period. 
	
\subsection{In-Sample Analysis} \label{emp_ins}
The aim of the in-sample analysis is to examine how effective dimensionality of the market correlation structure evolves across different market conditions. The lower-spectrum indicator provides the empirical proxy for this evolution.
We use the S\&P500 index as a reference indicator of the market (Sections \ref{esp1} and \ref{esp3}). To support the robustness of our results, the experiment is replicated on the NIKKEI index (Section \ref{esp4}).

\subsubsection{Experiment 1: SPsector, $(n = 10, w = 20)$}\label{esp1}
			
For the first experiment on the SPsector database, we set $w = 20$.  
\begin{figure}[H]
	\begin{center}
\includegraphics[width=0.79\linewidth]{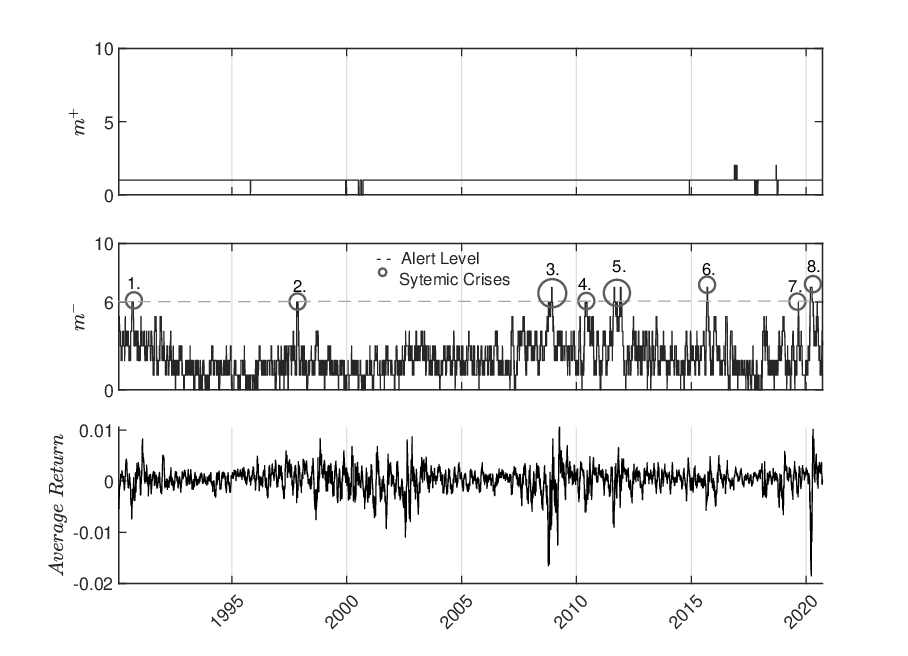}
\caption{Daily average returns of the S\&P500 index from 1990 to 2020 (bottom panel); $m^-$ (central panel) and $m^+$ (top panel) computed on SPsector dataset, $w=20$.  
Circles highlight the peaks of $m^-$ associated with periods of market distress. Global events according to \cite{sumer2023world}: 1. Gulf War or the US Savings and Loans Crisis (1989-1991), 2. East Asian Economic crisis (1997), 3. Subprime crisis (2007-2009), 4. and 5. European Debt crisis (2010-2012), 6.  Greece default to the IMF (2015), see \cite{laeven2020systemic}, 7. to be defined, probably the presence of high volatility in the money market as explained in a FEDS Note, see \cite{anbil2020happened}, 8. COVID-19 crisis. 
}\label{fig1}
\end{center}
\end{figure}

Figure \ref{fig1} depicts the time evolution of the lower-spectrum indicator (central panel) together with the daily average returns of the market index (bottom panel). This allows us to relate changes in the effective dimensionality of the market correlation structure with the events of market distress. 
The measure $m^+$ (top panel) is used as a benchmark. First, we observe that the value of $m^-$ ranges between $0$ and $7$. Small values of $m^-$ correspond to the periods where the market index, SPsector, is stable or grows moderately.

Conversely, high values of $m^-$ indicate that the numerical nullity of the correlation matrix is increasing, hence a significant percentage of the numerical rank of the correlation matrix is reduced. Consequently, the correlation structure approaches a lower-dimensional configuration, reducing the number of effectively independent investments and, therefore, the diversification opportunities. 
The empirical evidence is consistent with the theoretical interpretation developed in Section \ref{theo}.
During the most severe market events, the numerical rank of the correlation matrix decreases from its maximum value of 10 to 3. Consequently, although the market consists of ten sector portfolios, its correlation structure behaves as if only about three effectively independent investments were available, implying a significant reduction in diversification opportunities. 
The historical maximum points of $m^-$ can be compared ex-post with the major events of market distress documented in the literature. In the period under analysis, 
the largest values of $m^-$ are observed during several of the major events of market distress described by \cite{sumer2023world} and \cite{laeven2020systemic}: the recession of the early 90s, probably related to the Gulf war, and US Savings and Loans Crisis (1989-1991); the Asian financial crisis of 1997; the 2007-2009 subprimes crisis; the European debt crisis of 2010; the COVID-19 pandemic.

Observe that we do not have any widely recognized global event in correspondance to the peak identified by the circle $7$ in Figure \ref{fig1}, central panel. Indeed, the proposed approach is unable to discriminate the direction of the market when the correlations increase. Then, we cannot  exclude that a large majority of the assets correlate in a positive direction, as in the case of a speculative bubble. This limitation is expected, since the proposed indicator depends only on the correlation structure and not on the sign of market returns.
An alternative possibility is that some market configurations characterized by a similar correlation structure are actually not associated with widely recognized episodes of financial distress. Indeed, the measurement objectively associates a natural number with a market configuration, described in terms of lower-spectrum of the correlation matrix. 

For comparison purposes, we finally investigate the high part of the spectrum. This comparison highlights the different information conveyed by the extreme parts of the spectrum.
Observe that $m^+$ only takes the values $0$, $1$ and $2$. In the large majority of the periods $m^+ = 1$. This shows that, very often, only one eigenvalue of the correlation matrix is larger than $\lambda_{max}$. Many authors refer to the largest eigenvalue of the correlation matrix as the \textit{market mode}, see, for instance, \cite{Garlaschelli}. The most interesting aspect is that the few situations in which $m^+ \neq 1$ are not correlated with any particular market behavior. Even if $m^+$ is a proper measure of connectedness, it does not provide useful information, at least in the present application. 
	
\begin{figure}[h!]
	\centering
	\begin{minipage}{1\linewidth}
	\centering
	\includegraphics[width=0.48\linewidth]{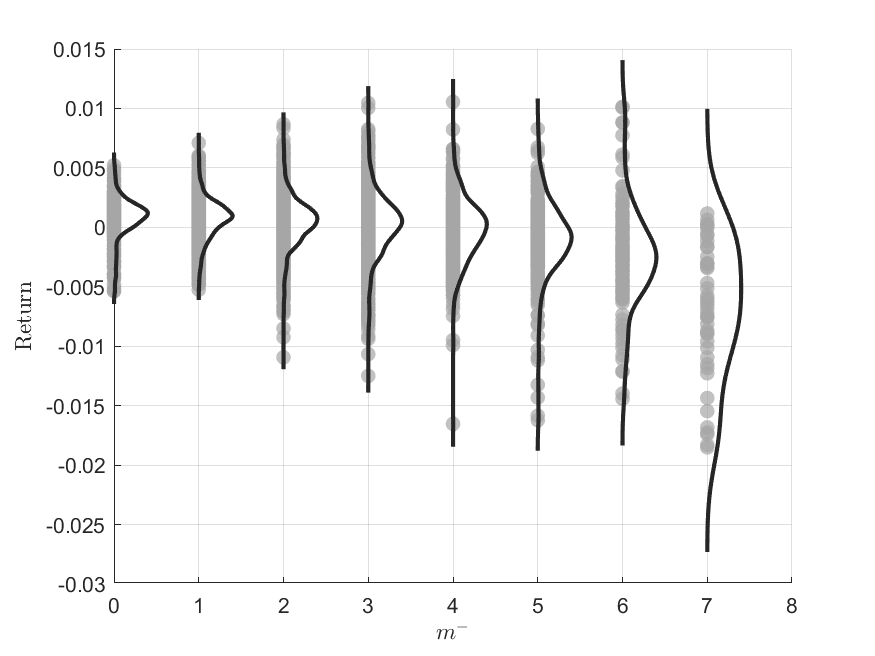}
	\includegraphics[width=0.48\linewidth]{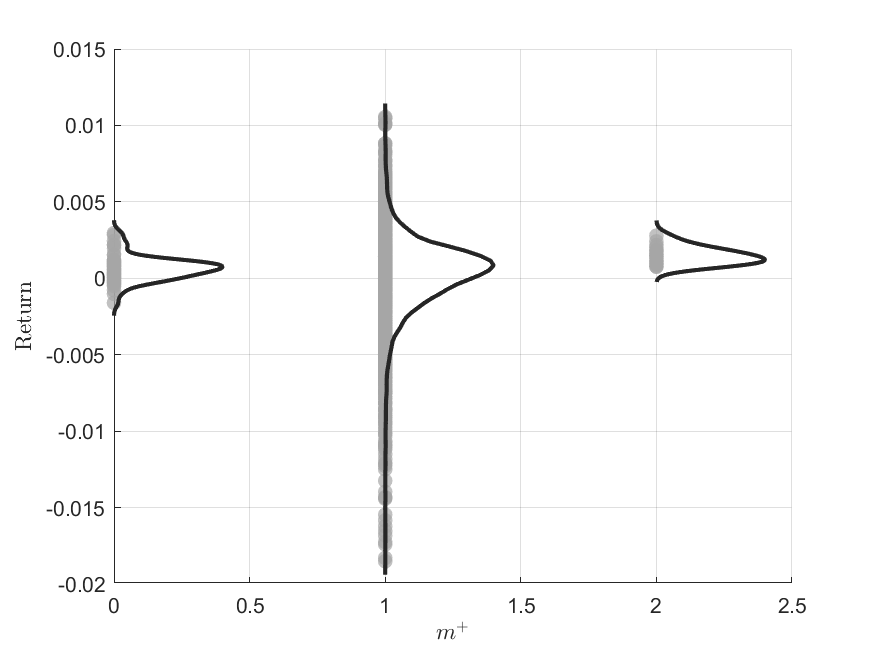}
	\end{minipage}
	\caption{SPsector, $w = 20$, distributions of the index daily average returns conditioned to the values of $m^-$ (left panel) and $m^+$ (right panel).}\label{fig2}
\end{figure}

To provide empirical support to the previous interpretation, we analyze the densities of the distribution of the index average daily returns conditioned to the values of 
$m^-$ (Figure \ref{fig2}, left panel). The distributions are obtained using  MatLab Gaussian kernel smoothing function \textit{ksdensity}, see \cite{peter1985kernel}. 
As the value of $m^-$ increases, the mean of the conditional distributions of returns decreases and the variance increases. This reinforces the interpretation that larger values of the lower-spectrum indicator correspond to market configurations characterized by a substantial reduction in effective diversification opportunities.
As a benchmark, we report the similar analysis by using $m^+$ (right panel), but no clear relationship with the returns of the referring index can be detected.

As already pointed out, the analysis at the sectoral level allows us to investigate the behavior of lower-spectrum indicator across various economic sectors. However, the shortcoming of using the sector portfolios instead of the single stocks is that their returns present a smoother behavior and a less variegate correlation structure, due to their positive and high correlations. 
To have a more realistic feedback about the proposed approach, it is necessary thus extending the analysis in a high dimensional framework, as we will show in the following subsection.

\subsubsection{Experiment 2: SPstocks $(n = 377, w = 20)$}\label{esp3}

In this experiment, we consider the SPstocks dataset to assess whether the proposed indicator still remains informative in a high-dimensional setting. Then, the size of the set of constituents passes from $n=10$ to $n=377$. Maintaining the same ratio $\frac{w}{n}$ as in the previous experiment, described in Section \ref{esp1}, requires setting $w = 754$. In this way, the Marchenko–Pastur thresholds are unchanged and the proposed methodology remains mathematically consistent.\\ 
Its empirical informativeness, however, depends on the compatibility between the estimation window and the temporal scale of the market events under investigation. 
Indeed, the resulting estimation window spans three years of daily observations, reducing the temporal response of the lower spectrum indicator to significant market events.
\begin{figure}[h!]
	\centering	
	\includegraphics[width=0.7\linewidth]{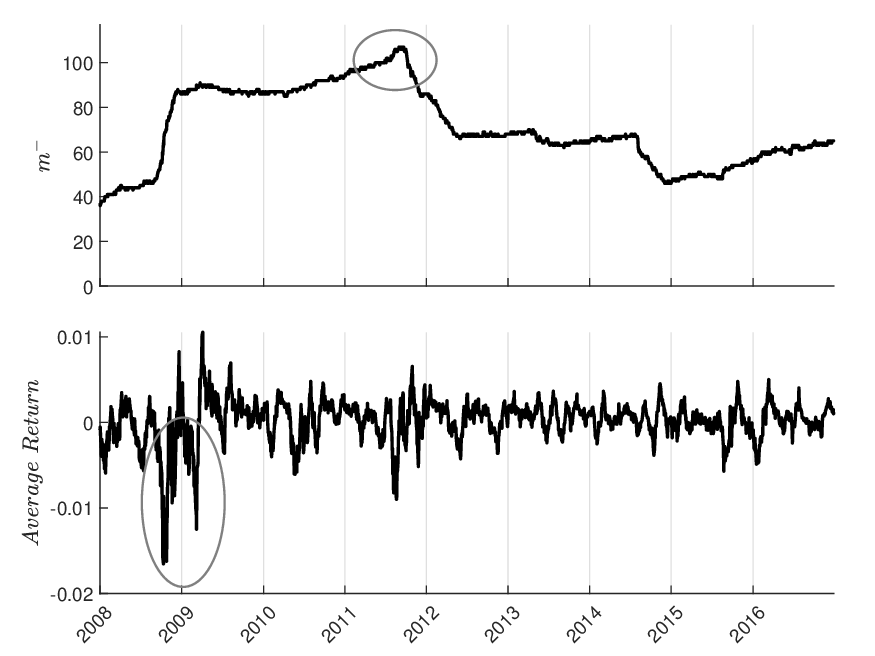}
	\caption{Daily average returns of the S\&P500 index from 01-03-2008 to 30-10-2016 (bottom panel); $m^-$ 
	(top panel) computed on SPstocks, $w = 754$.}\label{fig3}	
\end{figure}
 
 This is evident from the inspection of Figure \ref{fig3}, that compares the market index-- the average returns of the S\&P index -- with the measure $m^-$. The main peak of $m^-$
is observed in 2011, long after the onset of the 2007–2009 financial crisis. This delayed response is the natural consequence of estimating the indicator over such a long rolling window.
More precisely, the application of the proposed methodology to high-dimensional datasets requires satisfying two complementary conditions. On the one hand, the estimation window must be sufficiently short to follow the temporal evolution of the market. On the other hand, the resulting correlation matrix must preserve the theoretical interpretation of the lower spectrum.

To meet the first condition, we set $w=20$. However, for a high-dimensional dataset, as the case of the SPstocks, this implies that $n>w$. Then, the correlation matrix necessarily becomes rank-deficient, and a large number of eigenvalues are zero by construction. As a consequence, the structural zero eigenvalues become indistinguishable from the small eigenvalues associated with the numerical nullity, and the theoretical interpretation of $m^-$ is not longer valid. 
We therefore need to restore the conditions under which the lower part of the spectrum contains only the eigenvalues associated with the numerical nullity. The methodology proposed in Section \ref{theo} must be adapted so that the lower-spectrum indicator reflects only the numerical nullity induced by the correlation structure. In this way, the lower-spectrum indicator recovers its theoretical interpretation developed in Section \ref{theo}.
This can be achieved by reducing the dimensionality of the space of constituents. To this end, we group
the constituents into $k$ clusters adopting one among the standard clustering procedures - where we fix $k$ as an arbitrary value\footnote{For a direct comparison with the experiment in Section \ref{esp1} we set $k = 10$} such that $w >k$ and then by choosing $k$ portfolios each representing the constituents of its respective cluster. We then construct the matrix ${T \times k}$ in which columns are the returns of the aforementioned $k$ portfolios and the resulting reduced representation satisfies the inequality $w>k$.

Operatively, this reduction can be obtained through an agglomerative hierarchical clustering procedure\footnote{Hierarchical clustering is implemented through MATLAB's linkage and cluster functions. These routines construct a hierarchical cluster tree (dendrogram), from which a partition into $k$ clusters is obtained}. The assets are partitioned into $k$ clusters, and one representative portfolio is associated with each cluster. The flexibility of this process allows us to select the clustering level, or scale, that best fits the requirements of the application (in our case, $k=10$). This yields a reduced correlation matrix on which the lower-spectrum indicator can be computed according to the methodology developed in Section \ref{theo}.

As a further step we need to investigate whether the proposed cluster procedure also maintain its empirical effectiveness.
\begin{figure}[H]
	\begin{minipage}{1\linewidth}
		\centering
		\includegraphics[width=0.48\linewidth]{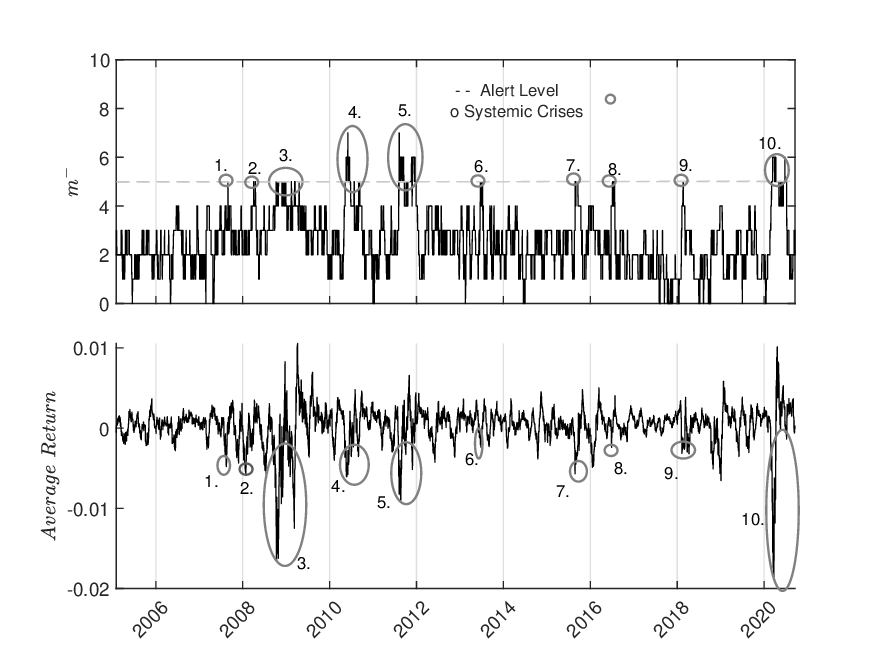}
		\includegraphics[width=0.48\linewidth]{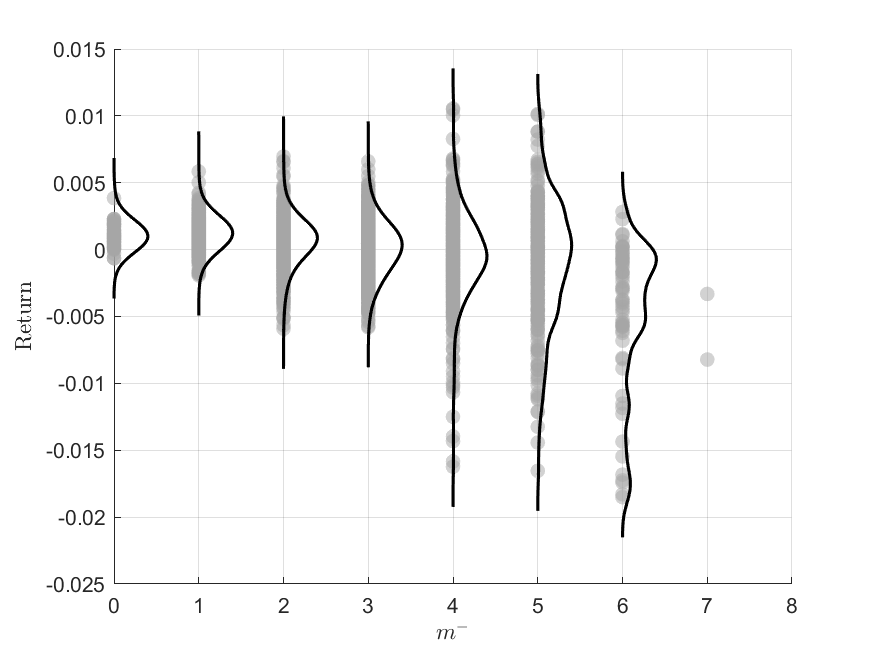}
	\end{minipage}
	\caption{Left panel: SPstocks, $w=20$. Average daily returns of the S\&P500 index and $m^-$ computed on after the reduction through hierarchical cluster tree procedure (number of clusters $k=10$).  Identifiable global crises during the period: 1. 2. and 3. Global Financial crisis 2007/09, 4. and 5. European Debt crisis, 
		7. Greece default in 2015, see \cite{laeven2020systemic}, 8. Brexit shock in 2016, see \cite{qiao2021brexit}, 
		10. COVID-19 crisis. Negative spikes n. 6 and n. 9 are not associated with known systemic events. 
		Right panel: scatter plot w.r.t $m^-$ and the correspondent conditional distributions of returns.}\label{fig5}
	
\end{figure} 
 Figure \ref{fig5} reports the results obtained for the SPstocks dataset with $w=20$. What emerges is that, even after this reduction, the peaks of $m^-$ are associated with the major events of market distress. Using the same interpretation as in the Experiment $1$ (Section \ref{esp1}), and setting the critical threshold to $5$
the measure intercepts the following systemic events: the Global Financial Crisis (2007-2009), the European Debt Crisis (2010-2012), the Greece default (2015), the Brexit effect (2016), and the COVID-19 crisis (2020).
As in the Experiment 1, not every peak can be unequivocally associated with a widely recognized systemic event. This is not in contrast with the interpretation of the indicator, as it measures changes in the correlation structure rather than predefined crisis labels.

Since the financial literature  traditionally focuses on the information carried by the largest eigenvalues of the correlation matrix, it is natural to compare the proposed methodology with an analogous construction based on the upper spectrum, by computation of the measure $m^+$. 
However, as in the previous case, its computation requires a reduction of the dimensionality of the set of constituents\footnote{We underline that the upper part of the correlation matrix spectrum is independent from the implemented dimensionality reduction procedure, as proved in the \ref{appendixA}. Nevertheless, the value of $m^+$ depends on the value $\lambda_max$ that is a function of matrix  dimension.}.
This is not unusual in the literature, and
a natural method to obtaining this reduction is to perform a Principal Component Analysis (PCA) on the original matrix.
Notice that the two indicators require different dimensionality-reduction methods because they rely on complementary parts of the spectrum.

Inspired by the work of \cite{ortobelli2015} we follow this approach to reduce the dimension of the data in order to compute $m^+$. We select the first $k$ principal components of the correlation matrix \footnote{In line with what previously done, we set $k=10$} and therefore we build the corresponding principal portfolios (for a definition of principal portfolio the reader can refer to \cite{Meucci}). We then compute their historical returns, and estimate a linear $k$-funds separation model, where the regressors are the returns of the $k$ principal portfolios (for more details see \cite{ross}). The returns of the original matrix are then approximated using the linear factor model. 
It is important to note that this procedure, as it is based on PCA, focuses only on the first $k$ eigenvalues. Moreover, the procedure of reduction to the first $k$ principal portfolios, and the subsequent reconstruction of the assets as their linear combination implies that the remaining $n-k$ eigenvalues are zero. The mathematical construction of the procedure, together with the proof that the first $k$ eigenvalues are preserved, is reported in \ref{appendixA}. \\
 In Figure \ref{fig4} we report the values of the measure $m^+$ computed by applying the procedure previously described. Compared with Figure \ref{fig5} (left panel), what is evident is that a few eigenvalues are above the upper bound defined by the Marchenko-Pastur benchmark. This is the sign that this measure weakly react to the systemic events. 
In conclusion, the information carried by the lower spectrum is substantially different from that contained in upper part. In this application, the lower-spectrum indicator provides a substantially more informative description of systemic market events than the benchmark measure 
$m^+$.

\begin{figure}[H]
\centering

\centering
\includegraphics[width=0.7\linewidth]{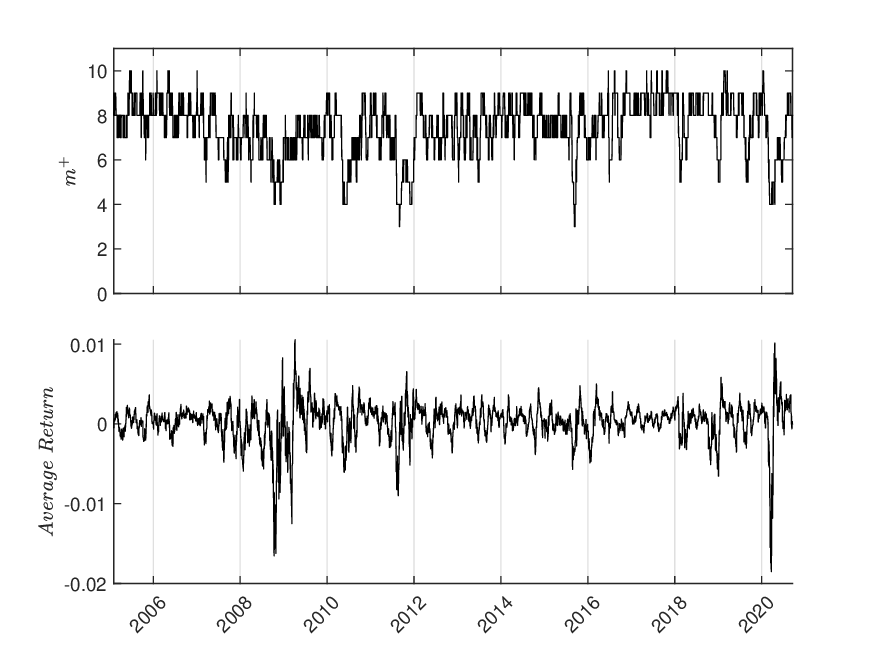}
\caption{SPstocks, $w=20$. Daily returns of the S\&P500 index (bottom panel) and the correspondent value of $m^+$ calculated the PCA procedure (top panel).}\label{fig4}

\end{figure}

\subsubsection{Experiment 3: Nikkei $(n = 199, w = 20)$}\label{esp4}

In this Section, we replicate the experiment using the Nikkei index. As in the SPstocks application, the dimensionality of the space of constituents is reduced in order to preserve the theoretical interpretation of the lower-spectrum indicator.

%
%
\begin{figure}[h!]
	
	\begin{minipage}{1\linewidth}
		\centering
		\includegraphics[width=0.48\linewidth]{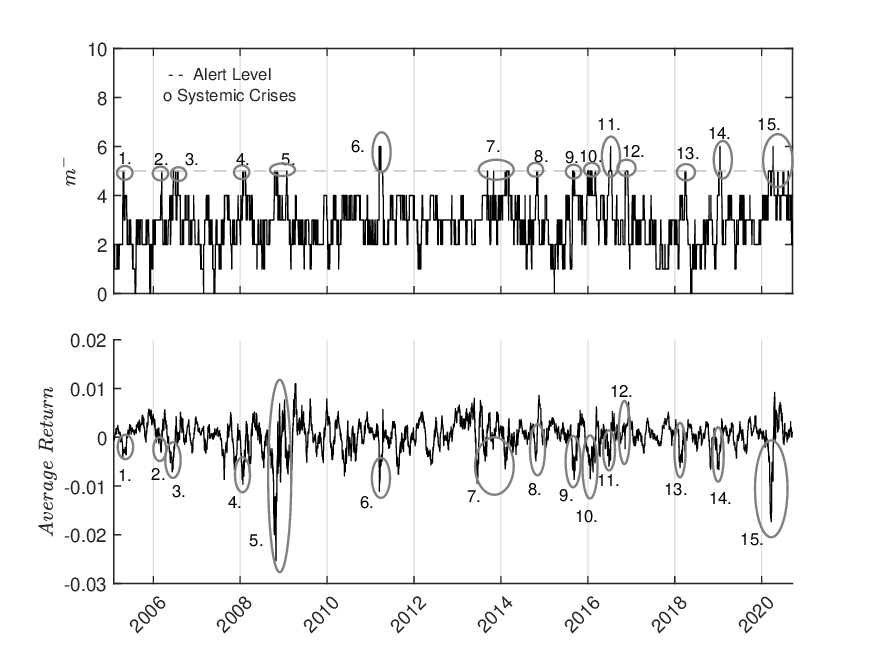}
		\includegraphics[width=0.48\linewidth]{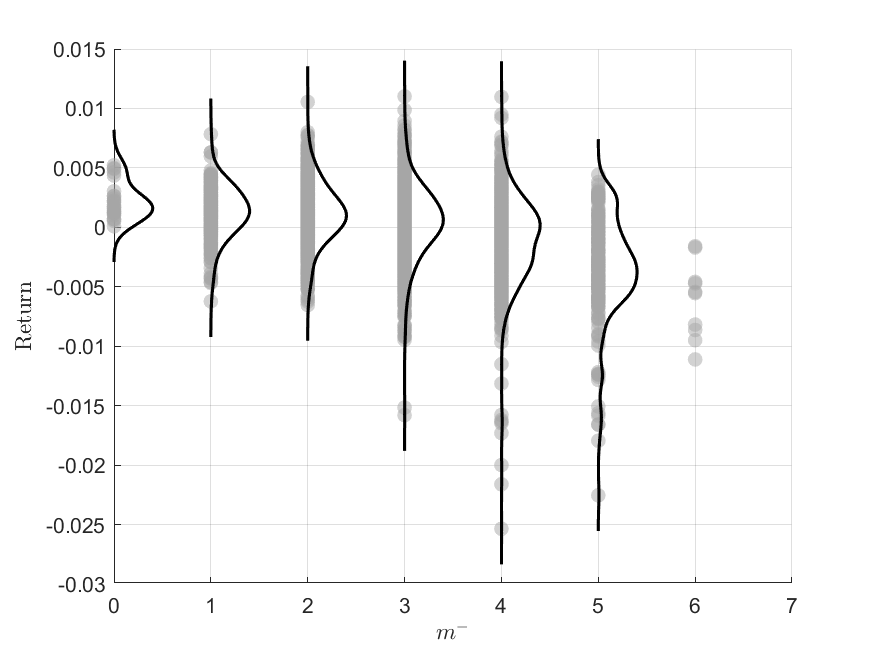}
	\end{minipage}
	\caption{Nikkei database, $w=20$. Daily returns of the NIKKEI index and the values of $m^-$ using hierarchical cluster tree procedure (left panel). Scatter plot w.r.t  $m^-$ and the correspondent conditional distributions of returns (right panel). Left panel systemic crises \cite{sumer2023world}: 4. and 5. Global Financial crisis 2007/09, 6. European Debt crisis 2011, 9. Greece default in 2015, see \cite{laeven2020systemic}, 11. Brexit shock in 2016, see \cite{qiao2021brexit}, 10. COVID-19 crisis. Unidentified events: peaks 1 - 3, 7, 8, 12, 13 and 14}\label{fig8}
\end{figure}

To this end, we set $w = 20$, and compute the measure $m^-$ after applying hierarchical clustering to select $k = 10$ portfolios. Results are reported in Figure \ref{fig8}. In line with the previous empirical experiments, we set the critical value equal to $5$.
The measure $m^-$ reacts to the major systemic events during the period under analysis: the subprime crisis of 2007-09, the European Debt crisis, the Greek default (2015),  the Brexit effect (2016), and the COVID-19 crisis (2020) (Figure \ref{fig8}, left panel). 
Some peaks are also observed during periods that are not immediately associated with recognized systemic events. This phenomenon, that we have previously detected also for the other experiments, is more relevant in this case than in the previous one. This may be partly explained by the low dimensionality of the index, in terms of number of constituents. Also, the index is not international enough to effectively represent a proxy of the global economy. As a consequence,
the events captured by the measure impact the majority of the index assets, but they remain 
confined to the Japanese market, without spreading to the global financial system. Nevertheless, the returns densities conditioned to the values of $m^-$, (Figure \ref{fig8} right panel), confirm the ability of the measure to identify the largest losses of the corresponding index.

\begin{figure}[H]
\centering
\begin{minipage}{1\linewidth}
	\centering
	\includegraphics[width=0.7\linewidth]{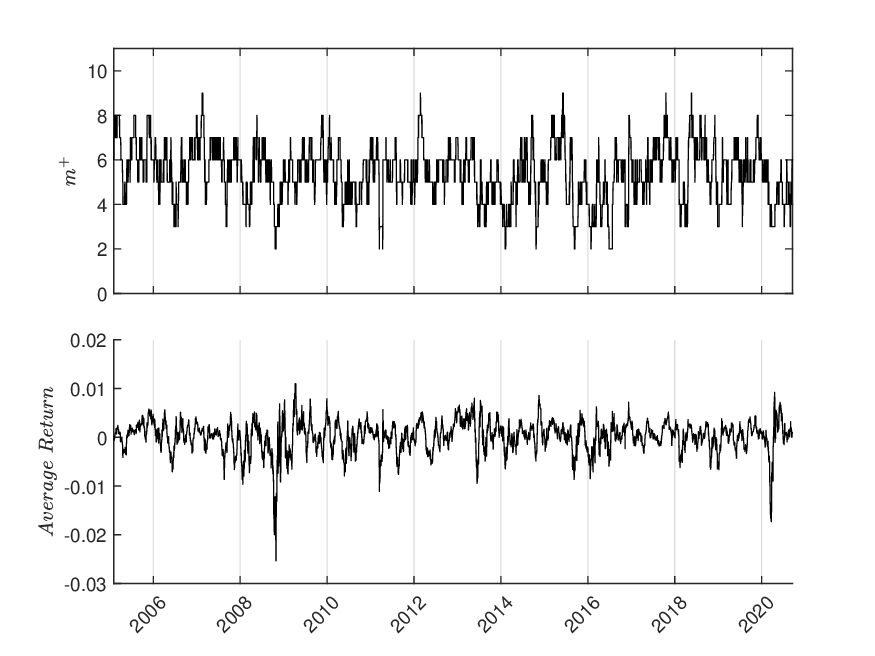}
\end{minipage}
	\caption{Nikkei database, $w=20$. Daily time series of NIKKEI index and $m^+$ using PCA procedure.}\label{fig7}
\end{figure}
For completeness, Figure \ref{fig7} reports the values of the benchmark measure $m^+$ computed through the PCA procedure described in the previous experiment. The behavior of $m^+$ is uninformative for detection of systemic events.

\subsection{Out-of-sample Analysis}\label{emp_oos}

Section \ref{emp_ins} established the descriptive validity of the lower-spectrum indicator. The next step is to investigate whether the information provided by the lower part of the spectrum can be exploited to anticipate future market conditions. This requires an out-of-sample analysis, where the indicator is used without relying on the information contained in the observations to be predicted.
The S\&P500 is still used as a reference indicator of the market, with the same datasets (SPsectors and SPstocks) described in Section \ref{data}.

The 
analysis is still conducted using a rolling-window framework, as described in Section \ref{data}. Specifically, given a window length $w$, the first $w$ observations of the matrix $A$ are used to compute the indicator. Then, starting from time $w+1$, $m^-$ is matched with the future returns of the financial index representing the economic system. 

As in the in-sample analysis, the choice of the correct window length deserves attention also in a predictive framework, and a suitable $w$ should be established. Indeed, $w$ must be large enough to ensure a stable estimate of both the eigenvalue distribution and the ratio $\lambda = \frac{n}{w}$, consistent with the theoretical framework, without reducing the ability of the measure to promptly react to changes in market conditions.
Let us recall that $m^-$ is computed in two steps. 
First, the spectrum is calculated starting from the matrix $A$, then the value of the measure is obtained by counting the number of elements of the spectrum below the interval of the Marchenko-Pastur theorem. In both steps the window length plays a central role, respectively identifying the number of rows extracted from $A$ and the threshold $\lambda$, which is needed to identify the interval.
 
The theoretical framework provides the methodology; however, its application to real financial data requires assessing if the theoretical threshold remains appropriate also in a predictive setting.
To this end, we consider in this analysis exogenous values of $\lambda$ fixed independently of the parameter $w$ to assess the robustness of the results.
For these reasons, in this experiment the past $w$ returns are used to compute the spectrum of $A$, while different values of the threshold $\lambda$ are imposed exogenously.

Figures \ref{fig9} and \ref{fig10} highlight this aspect by computing $m^-$ for $w = 20$ and for different values of $\lambda$, fixed independently of the parameter $w$, and reporting the box plots of the future (one-day-ahead) returns of the index (S\&P500 in this case) conditioned to the values of $m^-$. The top-left panels 
clearly show that $m^-$ has very limited out-of-sample discriminating power when computed with respect to the endogenous threshold; in this case, $\lambda = 0.5$ is the value implied by $w = 20$. Indeed, large losses (and also gains) almost uniformly correspond to all the possible values the measure can take.

\begin{figure}[h!]
\centering

\includegraphics[width=1\linewidth]{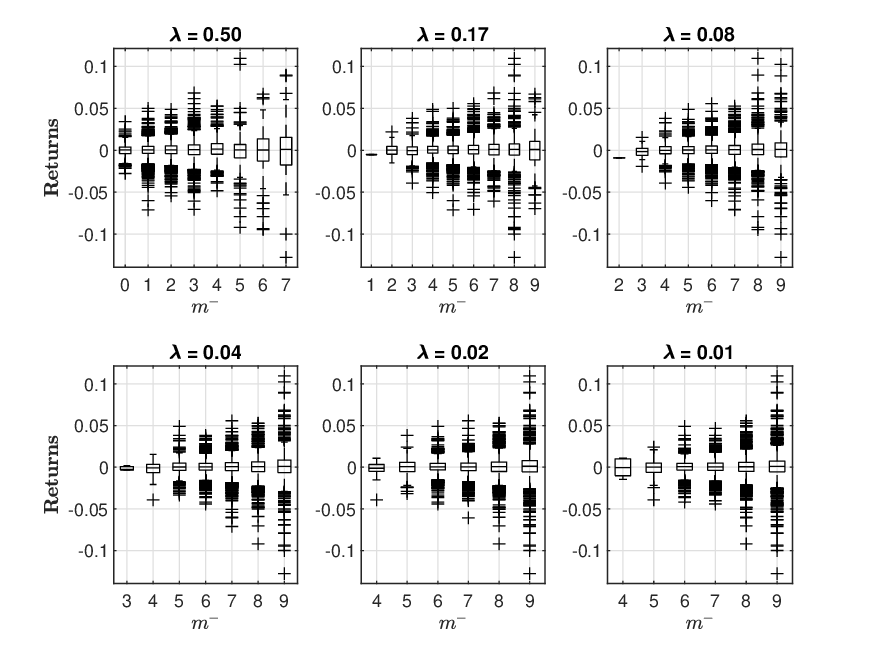}
\caption{SPsector database. Box plots with respect to $m^-$ and the corresponding conditional distributions of future (one-day-ahead) returns of the S\&P500 index.}\label{fig9}
\end{figure}

\begin{figure}[h!]
\centering

\includegraphics[width=1\linewidth]{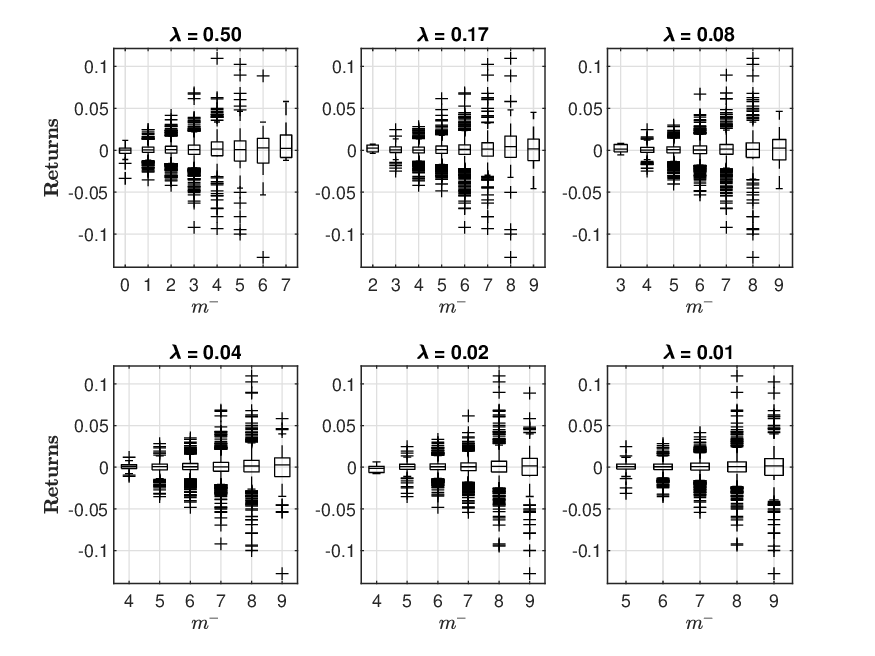}
\caption{SPstock database. Box plots with respect to $m^-$ and the corresponding conditional distributions of future (one-day-ahead) returns of the S\&P500 index.}\label{fig10}
\end{figure}

As shown in Figures \ref{fig9} and \ref{fig10}, for both the SPsector and SPstock databases, the effectiveness of the measure in forecasting large losses becomes clearer for small values of $\lambda$. Indeed, returns decrease when $m^-$ increases, then the predictive performance of $m^-$ improves when using $w = 20$ combined with a value of $\lambda$ that is smaller than the theoretical one. As graphically evident, the best results are obtained for $\lambda = 0.01$ (see the bottom-right panels of Figures \ref{fig9} and \ref{fig10}), where large losses (and also large gains\footnote{Let us underline how, in a predictive framework, $m^-$ seems to be less able to discriminate between large positive and negative gains than in a descriptive context; see Section \ref{emp_ins}.}) correspond to large values of the measure. This result can be explained by the fact that, due to the non-i.i.d. nature of real financial data, the effective value of $\lambda$ is likely to be smaller than the theoretical one. 

In principle, reducing $\lambda$ would require increasing the window length $w$. However, in a forecasting context, a larger $w$ implies estimating the correlation matrix over a longer window, assigning excessive weight to older observations and reducing the responsiveness of the measure to recent market conditions. This trade-off is a standard issue in out-of-sample applications: good predictions require primarily relying on more recent data, while the statistical and mathematical properties of the estimators call for longer time series. 
To address this issue, we introduce 
a decay factor that assigns larger weights to more recent observations and progressively smaller weights to older data. This is a common approach in finance, where
the decay factor quantifies the rate at which information exposures diminish in effectiveness (see \cite{metrics1996jp}). In our application, the decay factor can be interpreted as a persistent parameter following an exponential weighting scheme, where past observations are discounted at a constant rate. Formally, observations receive an exponential weight proportional to a given $\delta$, with $0 < \delta < 1$, which is the decay factor\footnote{The decay factor $\delta$ is the general term of the geometric series $(1-\delta)\sum_{t=0}^{\infty}\delta^t$.}
. Specifically, each observation is weighted by $\delta^t$.

This approach allows us to maintain a long window length $w$, i.e., a small value of the endogenous $\lambda$, while preserving sensitivity to recent information. The decay factor thus reconciles the theoretical requirements of the model with the empirical features of financial data in a forecasting context.

\begin{figure}[H]
	\centering
	\includegraphics[width=0.7\linewidth]{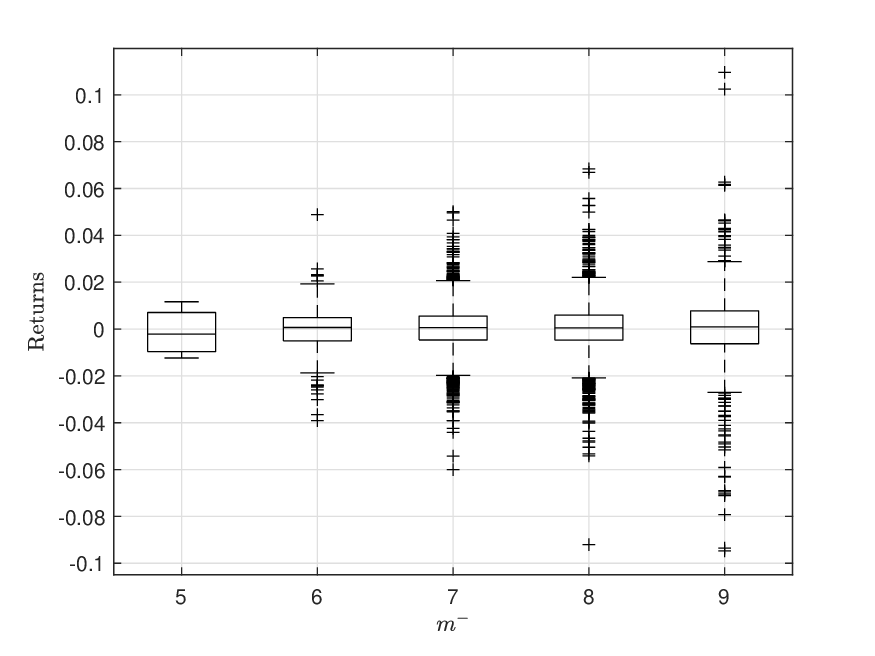}
	\caption{SPsector database. Box plot with respect to $m^-$ and the corresponding conditional distributions of future (one-day-ahead) returns of the S\&P500 index for $\delta=0.7$.}\label{fig11}
\end{figure}

Figure \ref{fig11} shows the results for $w=1000$, $\lambda=\frac{n}{w}=0.01$, and $\delta=0.7$. This value of $\delta$ corresponds approximately to using the most recent $20$ observations to estimate the spectrum. While maintaining the endogenous determination of the threshold, this configuration of parameters preserves the forecasting power of $m^-$, as is clear from Figure \ref{fig11}.

In the following, the experiment is extended to all the remaining databases. 
For each database, $n$ is equal to 10, or reduced to that value using the clustering procedure described in Section \ref{esp3}.

\begin{table}[h!]
	\centering
	\caption{Mean, standard deviation, Value at Risk at the 1\% significance level, and the cardinality of the classes ($N$) of future (one-day-ahead) return distributions conditioned on $m^-$ for different databases. The parameters are $w=1000$, $\lambda=\frac{n}{w}=0.01$, and $\delta = 0.7$.}
	\begin{tabular}{l l c c c c c }
\hline
&& $m^-$ & N & Mean&  Std & V@R1\%  \\
\hline
\multicolumn{7}{l}{SPsector} \\

 && 5  &     5  & -0.0015 &  0.0092 &   0.0124 \\
 && 6  &    352 & 0.0001 &  0.0095 &   0.0276 \\
 && 7  &   2341 & 0.0002  &  0.0099 &   0.0279 \\
 && 8  &   2367 & 0.0003  &  0.0115 &   0.0316 \\
 && 9  &    673 & 0.0001  &  0.0180 &   0.0676 \\
\hline
\multicolumn{7}{l}{SPstock} \\

&& 5 &   1   & -0.0022 &  -     &     -   \\
&& 6 &  109  & 0.0001 & 0.0081 &    0.0288 \\
&&7  &  791  & 0.0005 & 0.0095 &      0.0257\\
&&8  &  834  &0.0004 & 0.0116 &      0.0355\\
&&9  &  218  & 0.0004 & 0.0155 &     0.0453\\
\hline
\multicolumn{7}{l}{NIKKEI} \\
 && 6  &   40  & 0.0009 & 0.0095 &      0.0179 \\
  && 7 &   505 & 0.0000 &   0.0128 &        0.0323 \\
  && 8 &   994 &0.0003 &   0.0147 &      0.0359 \\
  && 9 &   308 & 0.0011 &   0.0176 &   0.0416 \\

\hline
\end{tabular}
\label{tab:oos}
\end{table}
Table \ref{tab:oos} reports the main statistics (mean, standard deviation and value at risk (V@R) at 1\% significance level) of future (one-day-ahead) return distributions conditioned on the values of $m^-$. We group the future returns in dependence on the value of the measure; that is, we report the absolute frequency $N$ of each class respect to $m^-$. Interesting results are common across the different databases. In particular, higher 
levels of $m^-$ are associated with higher levels of both standard deviation and V@R. This fact supports the capacity of $m^-$ to predict future market turbulence.
Moreover, the cardinality $N$ of the classes identified by the possible values of the measure provides very useful information on the discriminating power of the measure.
As highlighted in Table \ref{tab:oos}, the maximum value of the measure corresponds to a relatively small number of observations. In other words, 
$m^-$ is able to efficiently identify risky situations discriminating them with respect to natural market fluctuations. To exemplify, suppose an agent wants to use the measure as an early warning indicator. In this case, a good measure is expected to prevent future losses while not providing the warning signal too often. Indeed, false alarms represent an opportunity cost, requiring hedging against future losses that will not occur. Thus, an ideal early warning system would prevent large losses while minimizing false signals. As already evident 
by inspection of Figures \ref{fig9}, \ref{fig10}, and \ref{fig11}, this is also confirmed by the fact that no notable relationship between the value of $m^-$ and the average return is observed. 
In fact, the present approach is designed for risk detection and is not intended to forecast market direction.

To further support the robustness of our findings, we compare the results with those obtained using the \textit{cumulative risk fraction (CRF)}, a systemic risk measure proposed in \cite{13.Billio}. Since this measure is computed starting from the spectrum of the returns matrix, as are the ones proposed in the present research, the comparison is direct and meaningful.
The CRF relies on the idea that, during crises, a small number of dominant principal components explain a large fraction of market variance. Hence, the risk measure is computed starting from the leading eigenvalues, following an intuition similar to that behind $m^+$.\footnote{Let $\sigma_1, \dots, \sigma_n$ denote the $n$ singular values of matrix $A$, enumerated in non-ascending order. The \textit{cumulative risk fraction} is defined as $CRF^k(A)= \frac{\sum_{j=1}^{k}\sigma_j^2(A)}{\sum_{j=1}^{n}\sigma_j^2(A)}$, for $k=1,\dots,n$. For more details see \cite{13.Billio}.}

Figure \ref{figcrf} (left panel), shows the scatter plot of the future (one-day-ahead) return distribution of the S\&P 500 index with respect to the values of CRF. Let us note that CRF is continuous and can theoretically take values in the interval $(0,1)$, while $m^-$ can only assume integer values. Therefore, to directly compare the two measures we discretize CRF by partitioning its effective codomain, in this case the interval $(0.85,1)$, into 10 intervals of equal length. Then, the forecasting power of CRF is graphically represented similarly to what was done for $m^-$; see Figure \ref{figcrf}, right panel.

	\begin{figure}[h!]
	
	\begin{minipage}{1\linewidth}
		\centering
		\includegraphics[width=0.48\linewidth]{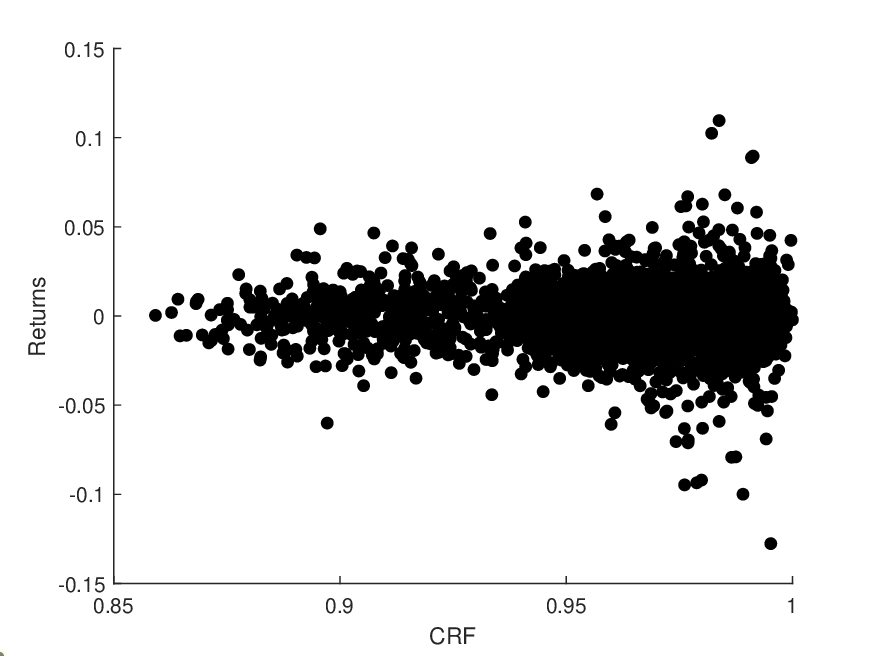}
		\includegraphics[width=0.48\linewidth]{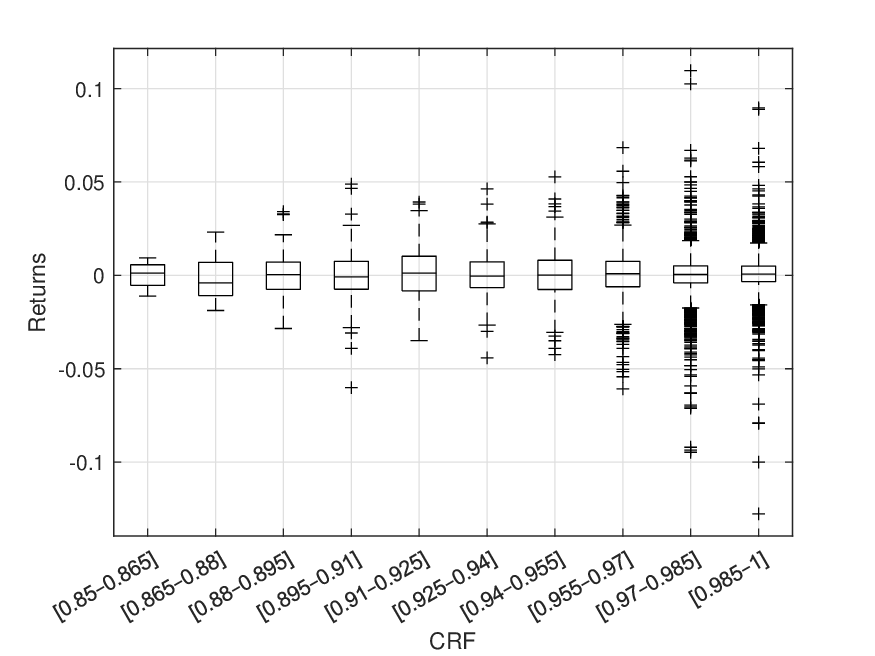}
		\end{minipage}
	\caption{SPsector database. Scatter plot (left panel) and Box plot (right panel) with respect to CRF and the corresponding conditional distributions of future (one-day-ahead) returns of the S\&P500 index.}\label{figcrf}
\end{figure}

Both panels of Figure \ref{figcrf} show that larger values of CRF correspond to large losses (and also large gains). Nevertheless, CRF is less effective in systemic risk detection than $m^-$. This depends on the fact that the classes with higher values of CRF contain a multitude of observations, not just a limited number of extreme returns. Therefore, the discriminating capacity of CRF is lower. 

\begin{table}[h!]
	\centering
	\caption{Mean, standard deviation, Value at Risk at the 1\% significance level, and the cardinality of the classes ($N$) of future (one-day-ahead) return distributions conditioned on CRF for different databases.}
	\begin{tabular}{l l c c c c  c }
\hline
&& $CRF$ & N & Mean& Std &  V@R1\%  \\
\hline
\multicolumn{7}{l}{SPsector} \\
&& [0.85–0.865]  &      4   &  0.0001 &    0.0085 &    0.0112\\
&& [0.865–0.88] &    26    & -0.0025 &     0.0109 &   0.0188\\
&&[0.88–0.895]  &    75   & -0.0003 &     0.0131&   0.0278\\
&&[0.895–0.91]  &   157 &   -0.0004&     0.0141&    0.0385\\
&&[0.91–0.925]  &   158  &    0.0014  &   0.0136 &    0.0313\\
&&[0.925–0.94]  &   140  &   0.0003 &    0.0135 &  0.0314\\
&&[0.94–0.955]  &   474  &    0.0000&     0.0129&    0.0305\\
&&[0.955–0.97]  &   902  &   0.0005&     0.0134  &  0.0363\\
&&[0.97–0.985] &   2716 &    0.0001 &    0.0115 &   0.0342\\
&&[0.985–1] &       3046 &    0.0005&     0.0102&    0.027\\
\hline
   \multicolumn{7}{l}{SPstocks} \\
   
&&[0.69–0.725] &    18   &     -0.0009&     0.0113&    0.0325\\
&&[0.725–0.76]&     96   &  -0.0014&   0.0105&    0.0274\\
&&[0.76–0.795]&    254   &  -0.0003 &   0.0156&    0.0503\\
&&[0.795–0.83]&   531   &  0.0005&   0.0123&    0.0369\\
&&[0.83–0.865]&    891 &   -0.0003&    0.0139&    0.0404\\
&&[0.865–0.9]&     832 &    0.0008&    0.009&      0.0287\\
&&[0.9–0.935]&     822 &    0.0003&    0.0107&     0.0319\\
&&[0.935–0.97]&    388 &    0.0010&    0.0127&    0.0359\\
\hline
\multicolumn{7}{l}{NIKKEI} \\

&& [0.5–0.55] &    45  &   0.0009&    0.0091&    0.0289\\
&&[0.55–0.6]&    238   &  -0.0010&     0.0136&    0.0398\\
&&[0.6–0.65]&    556   &   0.0012&     0.0116&    0.0279\\
&&[0.65–0.7]&    786 &   -0.0002&     0.0116&    0.0293\\
&&[0.7–0.75]&    873&    0.0001&     0.0130&    0.0386\\
&&[0.75–0.8]&    666&    -0.0003&     0.0160&    0.0498\\
&&[0.8–0.85]&    431 &    0.0005&     0.0182&    0.0575\\
&&[0.85–0.9]&    135 &     0.0018&     0.0221&    0.0640\\
&&[0.9–0.95]&     72 &     0.0012&     0.0362&     0.0476\\
&&[0.95–1]&        5  &    0.0074&     0.0262&    0.0166\\
\hline
\end{tabular}
\label{tab:billio}
\end{table}

Table \ref{tab:billio} reports the principal statistics summarizing the application of CRF to the other databases under analysis. Comparing Tables \ref{tab:oos} and \ref{tab:billio}, a few interesting evidences need to be underlined. First, the larger values of standard deviation and V@R do not generally correspond to the larger values of the measure. Let us underline that this fact does not depend on the discretization performed for the comparison. Moreover, the classes of the values of the measure containing larger losses are characterized by a large cardinality. This implies that the measure has limited discriminating power. In other words, implementing an early warning system based on CRF would generate many false signals, with a negative impact in terms of opportunity costs. Finally, as for $m^-$, the average return and the value of CRF show no notable relationship.

To summarize, $m^-$ performs well in a predictive setting, even when compared with CRF, an alternative similar measure proposed in the literature. Moreover, the results suggest that disregarding the smallest eigenvalues of the spectrum entails a loss of relevant information in a forecasting framework. 

\section{Conclusions}\label{conclusions}

This paper deal with the issue of interpreting the information contained in the lower part of the spectrum of financial correlation matrices. The proposed approach aims to be complementary to the traditional one, that looks at the upper part of the spectrum to extract relevant information about the market co-movement.  
The lower spectrum is interpreted through the concepts of numerical nullity, providing a rigorous mathematical interpretation that allows to connect spectral properties and the effective dimensionality of financial markets.
From a financial perspective, the proposed approach suggests that the lower spectrum provides a natural way to describe the progressive reduction of effective diversification opportunities associated with increasing market synchronization.
The interpretation of the lower part of the spectrum opens a new perspective on the way to address a fundamental financial phenomenon such as the market synchronization.
 
%
%

	\section*{Declarations}
	
	
%

	\noindent
	{\bf Data availability.}
	The datasets analyzed in the current study are available from the corresponding author on a reasonable request.
	
\appendix

\section{PCA procedure for upper part of the spectrum}\label{appendixA}

This appendix reports the formal construction of the PCA procedure for the upper part of the spectrum.

Let $V_k \in \mathcal{M}_{n \times k}$ be the matrix obtained applying the PCA to the matrix $S_A$ - that is the matrix with columns given by the first $k$ principal components obtained by the correlation matrix $\frac{1}{T} S_A'S_A$. Notice that by construction, $V_k$ is an orthonormal matrix, i.e. $V_k'V_k=I_k$.
Let us compute $\tilde{X} \in \mathcal{M}_{T \times k}$ by multiplying $S_A \in \mathcal{M}_{T \times n}$ and $V_k\in \mathcal{M}_{n \times k}$, that is $\tilde{X}=S_AV_k$. This means that we are expressing the assets returns as a linear combination of the $k$ principal components.\\
We then perform the regression $\hat{S}_A=\tilde{X}\beta$, where  $\beta=(\tilde{X}'\tilde{X})^{-1}\tilde{X}'S_A$. In this way we obtain a new matrix $\hat{S}_A \in \mathcal{M}_{T \times n}$, of rank $k$, representing the returns, in which every asset is estimated using the regressors of the $k$ principal portfolios.

The following result holds:\\

\begin{proposition}
The matrices $\frac{1}{T} S_A'S_A$ and $\frac{1}{T}\hat{S}_A'\hat{S}_A$ share the same first $k$ eigenvalues. 
\end{proposition}

\begin{proof}

Let $\lambda_1\geq\lambda_2\geq\cdots\geq\lambda_n$ be the eigenvalues of $\frac{1}{T} S_A'S_A$ and 
$\Lambda_k=\operatorname{diag}(\lambda_1,\ldots,\lambda_k)$ the diagonal matrix of the eigenvalues corresponding to the orthonormal matrix $V_k$. By definition of eigenvalues and eigenvectors, we have:
$$
\frac{1}{T} (S_A'S_A)V_k=V_k\Lambda_k
$$
Being $\tilde{X}=S_AV_k$, we have
$$
\tilde{X}'\tilde{X}
=
V_k'S_A'S_AV_k
=
T V_k'(\frac{1}{T} S_A'S_A)V_k
=
T V_k'V_k\Lambda_k
=
T I_k\Lambda_k
=
T\Lambda_k,
$$
and
$$
\tilde{X}'S_A
=
V_k'S_A'S_A
=
T V_k'\left(\frac{1}{T} S_A'S_A\right)
=
T\Lambda_kV_k'.
$$
where $V_k'\left(\frac{1}{T} S_A'S_A\right)=\Lambda_kV_k'$ being the matrix $\frac{1}{T} S_A'S_A$ symmetric. Therefore, by using the previous chain of equalities: 
$$\beta
=
(\tilde{X}'\tilde{X})^{-1}\tilde{X}'S_A
=
(T\Lambda_k)^{-1}T\Lambda_kV_k'
=
\Lambda_k^{-1}\Lambda_kV_k'
=
V_k'.
$$
that is the OLS coefficient matrix $\beta$ is equal to the transpose of matrix $V_k$ collecting the first $k$ principal eigenvectors.
It follows that
$$\hat{S}_A=\tilde{X}\beta=S_AV_kV_k'.$$
Recall that by construction, $rank(\hat{S}_A)=k$, then only $k$ eigenvalues of the matrix $\frac{1}{T}\hat{S}_A'\hat{S}_A$ are different from zero, whereas the remaining $n-k$ are zero.
Hence,
$$\begin{aligned}
\frac{1}{T}\hat{S}_A'\hat{S}_A
&=
\frac{1}{T}
V_kV_k'S_A'S_AV_kV_k' \\
&=
V_kV_k'V_k\Lambda_kV_k' \\
&=
V_kI_k\Lambda_kV_k' \\
&=
V_k\Lambda_kV_k'.
\end{aligned}
$$
Hence, the eigenvalues of
$\frac{1}{T}\hat{S}_A'\hat{S}_A$ are
$\lambda_1,\ldots,\lambda_k$, together with $n-k$ zero eigenvalues.
Therefore, $\frac{1}{T}S_A'S_A$ and
$\frac{1}{T}\hat{S}_A'\hat{S}_A$ share the same first $k$
eigenvalues.
\end{proof}			

%
%
%
%

\bibliographystyle{elsarticle-harv} 

						
						
						
						
						
					\end{document}